%% file: impact_effectiveness.tex
\newif\ifLNCS
\LNCStrue
\newif\ifnotLNCS
\notLNCSfalse

\ifLNCS
\RequirePackage{amsmath}
\fi
\documentclass{./llncs}
\ifLNCS
\usepackage{makeidx}  
\fi

\input{./macros}

\title{Defence Efficiency}

\author{Gleb Polevoy}
\institute{University of Amsterdam}

\begin{document}

\ifLNCS
\frontmatter          
\pagestyle{headings}  
\addtocmark{Price of Transition} 
\fi

\ifLNCS
\pagestyle{headings}  
\fi
\maketitle 



\input{abstract}

\input{introduction}
\input{model}

\input{conclusion}

\paragraph*{Acknowledgments}

This research is funded by the Dutch Science Foundation project
SARNET (grant no: CYBSEC.14.003 / 618.001.016)


\bibliographystyle{styles/splncs03}
\bibliography{library}


\end{document}

%% file: macros.tex
\usepackage[ruled]{algorithm2e}

\SetAlFnt{\small}
\SetAlCapFnt{\small}
\SetAlCapNameFnt{\small}
\SetAlCapHSkip{0pt}

\usepackage{url}

\usepackage{latexsym,amsmath,amssymb}
\ifnotLNCS
\usepackage{amsthm}
\fi
\usepackage{graphicx}

\usepackage{ifthen}
\usepackage{mdwlist,paralist}

\usepackage{comment}
\excludecomment{longversion}

\usepackage{epsfig}
\usepackage{epsf}

\usepackage[geometry]{ifsym}
\usepackage{rotating}

\usepackage{float}

\newboolean{Longversion} \setboolean{Longversion}{false}

\newcommand*{\eqns}[1]{Eq.~#1}
\newcommand*{\eqnsref}[1]{\eqns{\eqref{#1}}} 

\newcommand{\set}[1]{\left\{ #1 \right\}}

\newcommand{\reals}{\mathbb{R}}
\newcommand*{\realsP}{{\ensuremath{\mathbb{R}_{+}}}}

\newcommand*{\defas}{\ensuremath{\stackrel{\rm \Delta}{=}}}

\renewcommand{\qed}{\hfill $\blacksquare$}

\ifnotLNCS
\newtheorem{theorem}{Theorem}
\fi

\ifnotLNCS

\newtheorem{remark}{Remark}

\newtheorem{proposition}{Proposition}

\fi

\ifnotLNCS
\newtheorem{example}{Example}

\fi

\newcommand{\defined}[1]{\emph{#1}}

\renewenvironment{proof}[0]{{\it Proof. }}{\qed} 


%% file: abstract.tex
\begin{abstract}
In order to automate actions, such as defences against network attacks, one needs to quantify their efficiency.
This can subsequently be used in post-evaluation, learning, etc.
In order to quantify the defence efficiency as a function of the
impact of the defence and its total cost, we present several natural
requirements from such a definition of efficiency and provide a natural
definition that complies with these requirements. Next, we precisely
characterize
our definition of efficiency by the axiomatic approach; namely, we
strengthen the original requirements from such a definition and prove
that the given definition is the unique definition that satisfies those
requirements.
Finally, we generalize the definition to the case of any number of
input variables in two natural ways, and compare these generalizations.
\end{abstract}

%% file: introduction.tex
\section{Introduction}

Exact definitions and measurements are necessary for conducting science%
~\cite[Page~$1$]{Roberts1984}.
In particular, many actions, such as
defending against network attacks, defragmenting disks,
cleaning streets, and locating viruses,
would benefit from knowing their efficiency.
Knowing efficiency would allow for
improvements, including automatic improvements. We therefore
need to come up with a proper definition of the efficiency of actions.

Since we are not aware of a general definition of efficiency, we
define it in a natural way satisfying natural axioms, such as
monotonicity with respect to certain inputs.
This is a derived measure which inputs constitute a Cartesian
product~~\cite[Chapter~$5$]{Roberts1984}.
A famous problem of deriving measures is the arbitrary choice in
the definition~\cite{Adams1966}, and
there is some work on the existence and uniqueness of a derived
measurement of a certain form~\cite[Section~$2.5$]{Roberts1984}, but
we would like to have stronger uniqueness statements than homomorphisms.
Therefore, we ensure our definition is
the best for our goals by axiomatic characterisation.
We begin with an example of defending against a network attack, and then
we suggest a general formula which we characterise by the means of some natural axioms.
Next, we generalise this formula to any number of input parameters
in two possible ways. The first way is an expansion of the basic
formula, and we characterise this expansion by an expended set of axioms.
The second generalisation is a combination of the basic formulas.
We finally explore the relationships between these two natural
generalisations.

To summarise,
we suggest a natural uniquely characterisable way to define efficiency
and study its generalisations, allowing for a wide range of applications.

%% file: model.tex
\section{Network Defence}

Let us consider the scenario of defending against an attack in the
a network. 
Consider a network where a node is under attack, which reduces its
revenue. The implemented countermeasure can either bring upon a
recovery or not.

Denote the \defined{revenue} as a function of time be $r(t)$,
and assume that $r \colon \reals \to \realsP$.
Let the detection of the attack be $t_d$ and the recovery from it be
$t_r$; and let the baseline level be $B$. Then, the \defined{impact}
of an attack is defined as $I \defas \int_{t_d}^{t_r}{( B - r(t) )dt}$.
If no recovery takes place, then $t_r = \infty$, and then,
$I \defas \int_{t_d}^{\infty}{( B - r(t) )dt}$. This improper
integral can either converge or not.

Let the \defined{cost} of implementing a countermeasure be $c(t)$,
assuming that $c \colon \reals \to \realsP$.
Then, we define the \defined{total cost} as
$Ct \defas \int_{t_d}^{t_r}{ c(t) dt}$.
As the case with revenue, if no recovery is achieved, this
integral is improper and it can either converge or not.

In practice, we measure the revenue and the costs of the countermeasure
for a given time bound $T$; in particular, all the integrals are taken
at most till $T$. If no recovery takes place within
this period, we call the case not to have recovered and take the integral till $T$,
instead of infinity. This boundedness of time allows
us work with only proper integrals (the revenue and cost functions are
assumed to be bounded anyhow). Thus, the impact and the total cost are
always finite, regardless whether we consider a recovery to take place.


\section{Single Revenue and Cost}

We now model the situation and define the efficiency of a countermeasure
for the scenario above, though this can be applied to infinitely many
practical situations, where the efficiency decreases in both numerical
inputs.

Let $C$ be an upper bound on the cost
during the period $[0, T]$.
Preparing to define the \defined{efficiency} of a countermeasure, and we
require it to have the following properties:
\begin{enumerate}
	\item Monotonously decreasing with impact $I$, where $I \in [0, B \cdot T]$.
	\item	Monotonously decreasing with total cost $Ct$, where $Ct \in [0, C \cdot T]$.
	\item	If no recovery takes place, the efficiency is always smaller
	than if a recovery does take place, regardless anything else.
	\item	All the values between $0$ and $1$ are obtained, and only they are.
	In the functional notation, efficiency is a function
	$E \colon \set{\text{recovered, not recovered}} \times \realsP \times \realsP \to [0, 1]$.
\end{enumerate}

From the infinitely many definitions of efficiency that fulfill all
the above properties, we propose the following one.
We define the efficiency as
\begin{eqnarray}
E(\text{recovered or not}, I, Ct) \defas
\begin{cases}
\beta + \alpha \frac{B \cdot T - I}{B \cdot T} + (1 - \beta - \alpha) \frac{C  \cdot T - Ct}{C  \cdot T}\\
= 1 - \frac{\alpha}{B \cdot T} I - \frac{1 - \beta - \alpha}{C \cdot T} Ct  & \text{Recovered},\\
\alpha (\frac{\beta}{1 - \beta}) \frac{B \cdot T - I}{B \cdot T} + (1 - \beta - \alpha) (\frac{\beta}{1 - \beta}) \frac{C  \cdot T - Ct}{C  \cdot T}\\
= \beta - \alpha \frac{\beta}{(1 - \beta) (B \cdot T)} I - (1 - \beta - \alpha) \frac{\beta}{(1 - \beta) (C \cdot T)} Ct & \text{otherwise},
\end{cases}
\label{eq:efficiency_def}
\end{eqnarray}
where parameter $\beta$ defines the
division point between recovery and no recovery
(we allocate $\beta$ of the total $[0, 1]$ scale to the case
of no recovery, and the rest is given to the case of recovery),
and parameter $\alpha \in [0, 1 - \beta]$ expresses the relative
importance of the impact w.r.t.\ the total cost.
The idea is
to combine the relative saved revenue $\frac{B \cdot T - I}{B \cdot T}$ with the relative saved cost
$\frac{C  \cdot T - Ct}{C  \cdot T}$,
and shift the recovered case in front of the non-recovered one.
The multiplication by $\frac{\beta}{1 - \beta}$ normalizes the efficiency
of no recovery to fit to $[0, \beta]$.

The expression $\frac{B \cdot T - I}{B \cdot T}$ can obtain all the values
in $[0, 1]$, as $I$ is in $[0, B T]$. The expression $\frac{C  \cdot T - Ct}{C  \cdot T}$ obtains
all the values in $[0, 1]$, as $Ct \in [0, C \cdot T]$.
Therefore,
the defined efficiency obtains the values in
$[\beta + 0, \beta + (1 - \beta)] = [\beta, 1]$ if a recovery
takes place, and the values in
$[0, \beta]$ otherwise. The continuity of the efficiency function implies that all the values in
these segments are obtained.

Practically, we should take $C$ to be the smallest known upper bound, because a
non-tight bound makes the efficiency seem larger, since $C \cdot T - Ct$
will then never become zero, even for a very costly countermeasure.

The following characterization theorem proves that \eqnsref{eq:efficiency_def}
is the unique definition of efficiency, if we require a stronger set
of properties.
\begin{theorem}\label{the:effic_charact}
Let $Ct$ obtain values in $[0, C \cdot T]$. Then,
\eqnsref{eq:efficiency_def} is the unique
definition of efficiency that satisfies the following set of
properties:
\begin{enumerate}\label{efficiency_char}
	\item \label{efficiency_char:dec_imp}	Linearly decreasing with impact $I$, where $I \in [0, B \cdot T]$.
	\item	\label{efficiency_char:dec_cost}	Linearly decreasing with total cost $Ct$, where $Ct \in [0, C \cdot T]$.
	\item	\label{efficiency_char:same_lin_coeff}	The ratio of the linear coefficient
	of the impact to the linear coefficient of the total cost is the same,
	regardless whether the recovery takes place or not.
	\item	\label{efficiency_char:range}		If no recovery takes place, all the values between $0$ and $\beta$
	and only they can be obtained; if a recovery does take place, then all
	the values between $\beta$ and $1$ and only they can be obtained.
\end{enumerate}
\end{theorem}
We remark that condition~\ref{efficiency_char:same_lin_coeff} implies
that the ratio of the linear coefficient
of the impact to the linear coefficient of the total cost expresses their
relative importance, regardless whether recovery takes place.
\begin{proof}
\eqnsref{eq:efficiency_def} is linearly decreasing with impact and with
total cost and condition~\ref{efficiency_char:same_lin_coeff} holds in a
straight-forward manner.
We have showed after the definition of \eqnsref{eq:efficiency_def} that
condition~\ref{efficiency_char:range} is fulfilled as well. It remains to prove
the other direction.

Let the formula for the case when \emph{a recovery is attained} be
$a - b \cdot I - d \cdot Ct$, for positive $b$ and $d$. This form follows
from conditions~\ref{efficiency_char:dec_imp} and~\ref{efficiency_char:dec_cost}.
For the minimum impact and total cost, $I = Ct = 0$,
we have the maximum possible efficiency of $1$,
implying that $a - b 0 - d 0 = 1 \Rightarrow a = 1$.
For the maximum impact and total cost, $I = B \cdot T$ and
$Ct = C \cdot T$, we have the minimum possible efficiency of $\beta$,
which means that $1 - b \cdot B T - d \cdot C T = \beta \Rightarrow
b B T + d C T = 1 - \beta$. Let $\alpha$ be $b B T$. The nonnegativity of $d C T$
and $b B T + d C T = 1 - \beta$ imply together that $b B T \leq 1 - \beta$, as required
from $\alpha$ in \eqnsref{eq:efficiency_def}. Moreover,
$b B T + d C T = 1 - \beta$ implies that $d C T = 1 - \beta - \alpha$.
To conclude, the efficiency is $1 - b \cdot I - d \cdot Ct$,
where $b = \frac{\alpha}{B T}$ and $d = \frac{1 - \beta - \alpha}{C T}$,
for $\alpha \in [0, 1 - \beta]$, as in \eqnsref{eq:efficiency_def}.

In the case of \emph{no recovery}, let the formula be
$a'- b' \cdot I - d' \cdot Ct$. By substituting $I = Ct = 0$ we
conclude that $a' = \beta$. By substituting $I = B T$ and $Ct = C T$,
we obtain $\beta - b' B T - d' C T = 0$, i.e.~$b' B T + d' C T = \beta$.
From condition~\ref{efficiency_char:same_lin_coeff} we have
\begin{equation*}
\frac{b}{d} = \frac{b'}{d'} \iff \frac{b}{b'}= \frac{d}{d'}.
\end{equation*}
These two equations, together with the proven above equality
$b B T + d C T = 1 - \beta$, imply that each coefficient gets
multiplied by $\frac{\beta}{1 - \beta}$, yielding
$b' = b \frac{\beta}{1 - \beta}$ and $d' = d \frac{\beta}{1 - \beta}$.
Together with the expression above for
$a'$, we obtain \eqnsref{eq:efficiency_def}.
\end{proof}

\section{Generalizations}\label{Sec:model:gen}

The work till now assumed two inputs to the efficiency, besides the
fact whether the system has recovered: the impact and the total cost.
However, in some cases, more input variables are relevant.
For instance, consider the case with multiple revenues.
We denote the $i$th revenue by $r_i(t)$, its baseline level by $B_i$ and the corresponding $i$th impact
$I_i$, i.e.~$I_i \defas \int_{(t_d)_i}^{(t_r)_i}{( B_i - r_i(t) )dt}$, where
$(t_d)_i$ and $(t_r)_i$ are the $i$th detection and recovery time,
respectively.
Further, denote the $i$th cost of a countermeasure by $c_i(t)$, and the
$i$th total cost by $Ct_i$,
i.e.~$Ct_i \defas \int_{(t_d)_i}^{(t_r)_i}{ c_i(t) dt}$.
We can also have various time bounds $T_i$ for the countermeasures for
various revenues, and we say that the system has recovered if all the
revenues have recovered.

In general, we may have various input variables of any nature, which
should have positive or negative influence on the defined efficiency.
We present two different natural generalizations of the work above
to multiple variables. First, we expand~\eqnsref{eq:efficiency_def}
to consist of multiple terms. The second natural option is to simply
combine equations of type \eqnsref{eq:efficiency_def}.

\subsection{Expanding the Equation}\label{Sec:model:gen:exp_eq}

We generalize \eqnsref{eq:efficiency_def} as follows.
\begin{enumerate}
	\item We allow the efficiency decrease in strictly increasing
	functions of possibly multiple factors.
	Such a factor~$x$, where the strictly increasing function is $f$,
	such that $f(0) = 0$, 
	appears in the formula as $\frac{f(X) - f(x)}{f(X)}$, where $X$ is an
	upper bound on $x$. This equation obtains all the values from $1$ to
	$0$, when $x$ ranges from $0$ to $X$.
	
	\item	The efficiency may also increase in strictly increasing
	functions of additional, possible multiple, factors.
	Such a factor~$y$, where the strictly increasing function is $f$,
	appears as $\frac{f(y)}{f(Y)}$, where $Y$ is an upper bound on $y$.
	This equation obtains all the values from $0$ to
	$1$, when $y$ ranges from $0$ to $Y$.
\end{enumerate}

The factors w.r.t.~to
increasing functions of which the the dependency is increasing
are denoted as $y_1, \ldots, y_m$, and the factors w.r.t.~to
increasing functions of which the the dependency is decreasing
are denoted as $x_{m + 1}, \ldots, x_{m + l}$. 
Then,
$\alpha$ generalizes to $\alpha_1, \alpha_2, \ldots, \alpha_{m + l - 1}$,
describing the importance of the given function of each factor,
and the efficiency is defined as follows (w.l.o.g., we assume here that the
efficiency is decreasing w.r.t.~at least one factor).
\begin{eqnarray}
E(\text{recovered or not}, y_1, \ldots, y_m, x_{m + 1}, \ldots, x_{m + l}) \defas \nonumber\\
\begin{cases}
\beta + \sum_{i = 1}^{m} {\alpha_i \frac{f(y_i)}{f(Y_i)}}\\
+ \sum_{j = m + 1}^{m + l - 1} {\alpha_j \frac{f(X_j) - f(x_j)}{f(X_j)}} 
+ (1 - \beta - \sum_{k = 1}^{m + l - 1}{\alpha_k}) \frac{f(X_{m + l}) - f(x_{m + l})}{f(X_{m + l})}\\
= 1 - \sum_{k = 1}^m {\alpha_k} + \sum_{i = 1}^{m} {\frac{\alpha_i}{f(Y_i)} f(y_i)}\\
- \sum_{j = m + 1}^{m + l - 1} {\frac{\alpha_j}{f(X_j)} f(x_j)} 
- \frac{(1 - \beta - \sum_{k = 1}^{m + l - 1}{\alpha_k}) }{f(X_{m + l})} f(x_{m + l})
  & \text{Recovered},\\

\sum_{i = 1}^{m} {\alpha_i (\frac{\beta}{1 - \beta}) \frac{f(y_i)}{f(Y_i)}}\\
+ \sum_{j = m + 1}^{m + l - 1} {\alpha_j (\frac{\beta}{1 - \beta}) \frac{f(X_j) - f(x_j)}{f(X_j)}} 
+ (1 - \beta - \sum_{k = 1}^{m + l - 1}{\alpha_k}) (\frac{\beta}{1 - \beta}) \frac{f(X_{m + l}) - f(x_{m + l})}{f(X_{m + l})}\\

= \beta - (\frac{\beta}{1 - \beta}) \sum_{k = 1}^m {\alpha_k}
+ \sum_{i = 1}^{m} {\alpha_i (\frac{\beta}{(1 - \beta) f(Y_i)}) f(y_i)} \\
- \sum_{j = m + 1}^{m + l - 1} \alpha_j {(\frac{\beta}{1 - \beta}) \frac{1}{f(X_j)} f(x_j)}
- (1 - \beta - \sum_{k = 1}^{m + l - 1} {\alpha_k}) (\frac{\beta}{1 - \beta}) \frac{1}{f(X_{m + l})} f(x_{m + l})
	& \text{otherwise}.
\end{cases}
\label{eq:efficiency_def:gen}
\end{eqnarray}
As before, $\beta$ defines the division point between recovery and no
recovery. The parameters $\alpha_i$ fulfill
$\sum_{i = 1}^{l + m - 1} {\alpha_i}$ is between $0$ and $1 - \beta$.

Analogously to the basic case (where $m = 0$ and $l = 2$), we can show
that this efficiency fulfills the following conditions.
\begin{enumerate}
	\item	Monotonously increasing with each $f(y_i)$, where $y_i \in [0, Y_i]$.
	\item Monotonously decreasing with each $f(x_j)$, where $x_j \in [0, X_j]$.
	\item	If no recovery takes place, the efficiency is always smaller
	than if a recovery does take place, regardless anything else.
	\item	All the values between $0$ and $1$ are obtained, and only they are.
	In the functional notation, efficiency is a function
	$E \colon \set{\text{recovered, not recovered}} \times \realsP^{m + l} \to [0, 1]$.
\end{enumerate}

We generalize Theorem~\ref{the:effic_charact} as follows.
\begin{theorem}\label{the:effic_charact:gen}
For $i = 1, \ldots, m$, let $y_i$ obtain values in $[0, Y_i]$, and
for $j = m + 1, \ldots, m + l$, let $x_j$ be in $[0, X_j]$. Then,
\eqnsref{eq:efficiency_def:gen} is the unique
definition of efficiency that satisfies the following set of
properties:
\begin{enumerate}\label{efficiency_char:gen}
	\item \label{efficiency_char:gen:inc_y_i}	Linearly increasing with each $f(y_i)$, where $y_i \in [0, Y_i]$.
	\item	\label{efficiency_char:gen:dec_x_j}	Linearly decreasing with each $f(x_j)$, where $x_j \in [0, X_j]$.
	\item	\label{efficiency_char:gen:same_lin_coeff}	The ratio of the linear coefficient
	of the function $f$ of any variable to the linear coefficient of the function $f$ on any other variable is the same,
	regardless whether the recovery takes place or not.
	\item	\label{efficiency_char:gen:range}		If no recovery takes place, all the values between $0$ and $\beta$
	and only they can be obtained; if a recovery does take place, then all
	the values between $\beta$ and $1$ and only they can be obtained.
\end{enumerate}
\end{theorem}
\begin{proof}
This theorem is proven analogously to Theorem~\ref{the:effic_charact},
besides proving that the conditions of this theorem imply the formula also
for the case of no recovery, after having proven the rest. We prove this
part now. Conditions~\ref{efficiency_char:gen:inc_y_i}
and~\ref{efficiency_char:gen:dec_x_j} allow us assume that the utility
for no recovery looks as
$a' + \sum_{i = 1}^m {d_i' f(y_i)} - \sum_{j = m + 1}^{m + l} {b_j' f(x_j)}$.
First, the maximal possible value for no recovery, $\beta$, is obtained
by substituting $Y_i$ for each respective $y_i$ and zeros for each $x_j$.
This substitution yields $a' + \sum_{i = 1}^m {d_i' f(Y_i)} - 0 = \beta$,
implying that $a' = \beta - \sum_{i = 1}^m {d_i' f(Y_i)}$.

The least possible value for no recovery is zero, and it is attained when
each $x_j$ is $X_j$ an each $y_i$ is zero. This provides
$\beta - \sum_{i = 1}^m {d_i' f(Y_i)} + 0 - \sum_{j = m + 1}^{m + l} {b_j' f(X_j)} = 0$,
implying $\sum_{i = 1}^m {d_i' f(Y_i)} + \sum_{j = m + 1}^{m + l} {b_j' f(X_j)} = \beta$.
Condition~\ref{efficiency_char:gen:same_lin_coeff} implies that for any
two variables, w.l.o.g., for $y_i$ and $x_j$ there holds%
\footnote{Analogously for $y_{i_1}$ and $y_{i_2}$, or for $x_{j_1}$ and $x_{j_2}$.} 
$$\frac{d_i}{b_j} = \frac{d_i'}{b_j'} \iff \frac{d_i}{d_i'} = \frac{b_j}{b_j'},$$
where $d_i$ and $b_j$ are the coefficients for the case of recovery.
Assuming we have proven the formula for the case of recovery, we know
that $\sum_{i = 1}^m {d_i f(Y_i)} + \sum_{j = m + 1}^{m + l} {b_j f(X_j)} = 1 - \beta$.
Since the ratios of the coefficients remain the same, but the sum in the
case of no recovery is $\beta$ instead of $1 - \beta$, we need to
multiply the coefficients of the case of recovery by
$\frac{\beta}{1 - \beta}$. In particular, the above equation
$a' = \beta - \sum_{i = 1}^m {d_i' f(Y_i)}$ implies that
$\beta - \sum_{i = 1}^m {\frac{\beta}{1 - \beta} \frac{\alpha_i}{f(Y_i)} f(Y_i)}
= \beta - \sum_{i = 1}^m {\frac{\beta}{1 - \beta} {\alpha_i}}$, completing
the proof.
\end{proof}

We now remark on making this generalization even more general.

\begin{remark}
Each variable $y_i$ and $x_j$ can be equipped with its own increasing
function, and this leaves the statements and their proofs unchanged.
\end{remark}

\subsection{Combining Equations}

The generalization we present now allows for several inputs, as
Section~\ref{Sec:model:gen:exp_eq} does, and furthermore, we can now
also allow for some revenues to have recovered while others may have not
recovered.

We use the following efficiency as a black box to define the efficiency
of the $i$th countermeasure, $E_i$.
\begin{eqnarray}
E(\text{recovered or not}, y, x) \defas \nonumber\\
\begin{cases}
\beta + \alpha \frac{f(y)}{f(Y)}
+ (1 - \beta - {\alpha}) \frac{f(X) - f(x)}{f(X)}\\
= 1 - \alpha + {\frac{\alpha}{f(Y)} f(y)}
- \frac{(1 - \beta - {\alpha}) }{f(X)} f(x)
  & \text{Recovered},\\

{\alpha (\frac{\beta}{1 - \beta}) \frac{f(y)}{f(Y)}}
+ (1 - \beta - {\alpha}) (\frac{\beta}{1 - \beta}) \frac{f(X) - f(x)}{f(X)}\\

= \beta - (\frac{\beta}{1 - \beta}) {\alpha}
+ {\alpha (\frac{\beta}{(1 - \beta) f(Y)}) f(y)} 
- (1 - \beta - {\alpha}) (\frac{\beta}{1 - \beta}) \frac{1}{f(X)} f(x)
	& \text{otherwise}.
\end{cases}
\label{eq:efficiency_def:base_plus_min}
\end{eqnarray}

The total efficiency is then defined
as follows:
\begin{eqnarray}
E \defas \sum_{i = 1}^n {\gamma_i E_i},
\label{eq:efficiency_def_gen}
\end{eqnarray}
where the nonnegative parameter $\gamma_i$ describes the importance of $i$th
revenue. By taking normalized $\gamma_i$s, such that the combination is
convex, meaning that $\sum_{i = 1}^n {\gamma_i} = 1$, we ensure that
$E$ is in $[0, 1]$, because all the $E_i$s are there.

Assuming that either all the revenues recover or not, as required by the
previous generalization, a natural question is about the connection
between this and the previous generalization of
\eqnsref{eq:efficiency_def}.
In general, the two generalizations are not equivalent, as the following
example demonstrates. The gist of this example is that linear combination
treats the ratios in the recovery and the non-recovery cases differently.
\begin{example}
For $i = 1, 2$, let the efficiency of the $i$th countermeasure be 
\begin{eqnarray*}
1 + \frac{\alpha_i}{Y} y - \frac{1 - \beta_i - \alpha_i}{X} x\\
\beta_i + (\frac{\beta_i}{1 - \beta_i}) \frac{\alpha_i}{Y} y - (\frac{\beta_i}{1 - \beta_i}) \frac{1 - \beta_i - \alpha_i}{X} x,
\end{eqnarray*}
for the case of recovery and no recovery, respectively.
Then, in the combined efficiency of the two countermeasures,
the ratio between the coefficients of $y$ and $x$ in the case of
recovery is 
\begin{equation}
- \frac{\gamma_1 \frac{\alpha_1}{Y} + \gamma_2 \frac{\alpha_2}{Y}}
{\gamma_1 \frac{1 - \beta_1 - \alpha_1}{X} + \gamma_2 \frac{1 - \beta_2 - \alpha_2}{X}},
\label{eq:rat_1_2_recovery}
\end{equation}
while in the case of no recovery, the ratio is
\begin{equation}
- \frac{\gamma_1 (\frac{\beta_1}{1 - \beta_1}) \frac{\alpha_1}{Y} + \gamma_2 (\frac{\beta_2}{1 - \beta_2}) \frac{\alpha_2}{Y}}
{\gamma_1 (\frac{\beta_1}{1 - \beta_1}) \frac{1 - \beta_1 - \alpha_1}{X} + \gamma_2 (\frac{\beta_2}{1 - \beta_2}) \frac{1 - \beta_2 - \alpha_2}{X}}.
\label{eq:rat_1_2_no_recovery}
\end{equation}
These are, generally speaking, not equal: for instance, by substituting
$\gamma_1 = \gamma_2 = 0.5$, $\alpha_1 = \alpha_2 = 0.5$,
$\beta_1 = 0.5, \beta_2 = 0.4$ and $Y = X$,
\eqnsref{eq:rat_1_2_recovery} gives $- 10$, but
\eqnsref{eq:rat_1_2_no_recovery} yields $- 12.5$. These ratios are
not equal, thereby violating
condition~\ref{efficiency_char:gen:same_lin_coeff} of
Theorem~\ref{the:effic_charact:gen}.
Therefore, expanding equations is not equivalent to combining equations.
\end{example}

However, when the system recovers, the two generalizations are equivalent,
as we prove next.
\begin{proposition}\label{prop:equiv_if_recov}
If all the system recovers, then expanding equations is equivalent
to combining equations.
\end{proposition}
\begin{proof}
Combining equations includes expanding an equation, because any
equation of the form \eqnsref{eq:efficiency_def:gen} can be obtained by
combining equations $\beta + (1 - \beta) \frac{f(y_i)}{f(Y_i)}$ and
$\beta + (1 - \beta) \frac{f(X_j) - f(x_j)}{f(X_j)}$ with the coefficients
$\gamma_i \defas \alpha_i, i = 1, \ldots, m + l - 1$.

In order to show that expanding includes combining,
we prove that combining equations of the form
\eqnsref{eq:efficiency_def:gen},
which includes \eqnsref{eq:efficiency_def:base_plus_min}, using
\eqnsref{eq:efficiency_def_gen} yet again yields an equation of the form
\eqnsref{eq:efficiency_def:gen}.
Consider the part of \eqnsref{eq:efficiency_def:gen} that refers to the
case when a recovery is achieved, and look at the expression before
the equality sign. The expression is $\beta$ plus a linear combination
of terms such as $\frac{f(y_i)}{f(Y_i)}$ and
$\frac{f(X_j) - f(x_j)}{f(X_j)}$, such that the sum of the coefficients
of these terms is $1 - \beta$.
Therefore, combining $n$ such equations, the $i$th equation having
$\beta_i$, according to \eqnsref{eq:efficiency_def_gen} will yield 
$\sum_{i = 1}^n {\gamma_i \beta_i}$ plus a linear combination of
terms such as $\frac{f(y_i)}{f(Y_i)}$ and
$\frac{f(X_j) - f(x_j)}{f(X_j)}$, such that the sum of the coefficients
of these terms is
$\sum_{i = 1}^n {\gamma_i (1 - \beta_1)} = 1 - \sum_{i = 1}^n {\gamma_i \beta_1}$,
which is exactly an expression of the type of
\eqnsref{eq:efficiency_def:gen} for the case of recovery.
\end{proof}

%% file: conclusion.tex
\section{Conclusion}

We first presented a basic efficiency model where we had two parameters,
characterised it axiomatically, and subsequently generalized it
in two natural ways. Then, we compare these two ways, showing that they
are generally not the same, but if a recovery takes place, then they
are the same. Basically, the characterisations and the partial equivalence
of the natural generalisations hint that there may be only one natural way
to approach efficiency.

We may look at another axioms and at other generalisations. It would be
nice to axiomatise the way the generalisation stands with respect to
the original formula.
We may also consider eliciting the paramters of the formulas.
While using the formulas in practice, we may want to be able to
recalculate the results after an update about the values of some
parameters arrive. In order to be able to perform that without storing
extra information, we may need to look for other formulas.

To conclude, we provide a natural definition of efficiency, the only
definition that fulfils a natural set of axioms.